\documentclass[letterpaper,twocolumn,preprintnumbers,floatfix]{revtex4}

\usepackage{amsmath, amsfonts, amssymb, amscd, latexsym}
\usepackage{url}
\usepackage{times}
\usepackage{graphicx}
\usepackage{program}
\usepackage{color}
\usepackage{rotating}
\usepackage{array}

\newtheorem{theorem}{Theorem}[section]
\newtheorem{lemma}[theorem]{Lemma}

\newcounter{thedef} 
\newenvironment{definition}[1][Definition]{\begin{trivlist}
\item[\hskip \labelsep {\bfseries #1 \arabic{thedef}} \stepcounter{thedef}]}{\end{trivlist}}

\newcommand{\jo}[1]{ }

\begin{document}

\title{High-performance Energy Minimization in Spin-glasses \\ with Applications to Adiabatic Quantum Computing}

\author{
\begin{tabular}{ccccc}
H\'{e}ctor J. Garc\'{i}a & \;\;\;\; & Igor L. Markov \\
\footnotesize \tt hjgarcia@eecs.umich.edu & &
\footnotesize \tt imarkov@eecs.umich.edu \\
\end{tabular}
}

\affiliation{\\}




\begin{abstract}
Energy minimization of Ising spin-glasses has played a central role in
statistical and solid-state physics, facilitating studies of phase transitions
and magnetism. Recent proposals suggest using Ising spin-glasses for
non-traditional computing as a way to harness the nature's ability to find
min-energy configurations, and to take advantage of quantum tunneling to boost
combinatorial optimization. Laboratory demonstrations have been unconvincing so
far and lack a non-quantum baseline for definitive comparisons. In this work we
{\em (i)} design and evaluate new computational techniques to simulate natural
energy minimization in spin glasses and {\em (ii)} explore their application to
study design alternatives in quantum adiabatic computers. Unlike previous work,
our algorithms are not limited to planar Ising topologies. In one CPU-day, our
branch-and-bound algorithm finds ground states on $100$ spins, while our local
search approximates ground states on $1,000,000$ spins. We use this
computational tool as a simulator to study the significance of {\em
hyper-couplings} in the context of recently implemented adiabatic quantum
computers.
\end{abstract}

\maketitle

\section{Introduction}  \label{sec:intro}

The Ising spin-glass model was first proposed by E. Ising in 1925
as a mathematical model to understand the dynamics of phase transitions in
ferromagnetic systems. Such systems are composed of particles that can be in
either of two possible energy {\em spin} states. These spins interact in pairs
to produce an energy landscape that describes the overall behavior of the
system. The model is described in graph-theoretic terms by representing atoms
in a crystal with vertices and bonds between atoms with edges. Despite its
simplicity, the model has become an essential research tool in the analysis of
different kinds of physical systems such as stiff polymers \cite{MullerF03} and
genome sequences \cite{BaranK06} which can be mapped exactly or approximately
to an Ising model. Since physical systems found in nature are often disordered,
unless cooled to $0\,^\circ\mathrm{K}$, the model incorporates
randomness---either in the realization of the atomic couplings or the spin
states of the atoms.
	
    \begin{definition}
        {\em Spin glasses} are solid materials in which the magnetic moments (spins) have
        disordered orientations and the strength of the nearest-neighbor spin interactions (bonds) 
	are randomly distributed. 
    \end{definition}

Physical and chemical properties of a crystal depend on the total
energy of the bonds, which depend on atomic states. In particular, the lower
the total energy, the harder the material. Estimating total energy by a
graph-based function facilitates the use of graph algorithms to study
properties of solids. For instance, if an Ising graph
represents a ferromagnetic system, then finding the {\em generating function}
of cuts in said graph is equivalent to calculating the distribution of physical
states over all possible energy levels. On the other hand, finding the
minimum-cut (max-flow) of an Ising-model graph that represents an amino-acid
sequence is equivalent to calculating the lowest-energy configuration for a
corresponding protein \cite{BakkH03}. For most Ising models commonly studied in the
Physics community, the latter problem is equivalent to finding the ground-state
energy of the underlying physical system. Formally, the {\em ground-state
determination problem} (GSD) is defined as follows. Given an instance of an
Ising-model graph, find the set of spin-state values or {\em spin
configuration} that minimizes the overall energy of the underlying physical
system described by the graph. Such a state is known as the {\em ground state}
of the system. Barahona \cite{Barahona83} proved that, for a general random-field model,
the GSD problem is NP-hard. Thus, all known algorithms for finding optimal
solutions exhibit super-polynomial runtime. \\


\noindent {\bf Computing based on energy minimization in physical systems}.
Given that many physical systems have a natural ability to find least-energy
states, researchers are currently attempting to exploit this phenomenon to
perform useful computation. At the atomic scale, in addition to high
bit-density, energy optimization can be aided by quantum tunneling,
which effectively reduces the number of local minima. Thus, GSD problems
are of particular interest to quantum-information researchers because they are
suitable candidates for evaluating the performance of {\em adiabatic quantum
computers} (AQCs).  Recently developed AQCs employ an architecture based on
Ising spin systems \cite{Dwave}. First, the spin system is configured to represent
a given combinatorial problem, i.e., the spin interactions are carefully controlled
rather than random as in spin glasses. The ground state is found via {\em quantum
annealing} (the quantum analogue of thermal annealing), then read off as a bit
sequence and interpreted as an answer to the combinatorial problem. While the
classical formulation of GSD is NP-complete, Oliveira and Terhal \cite{OliveiraT08}
proved that formulating a general GSD instance in the context of AQC is
QMA-complete\footnote{QMA-complete is the quantum analogue of NP-complete.}.
Since the complexity of GSD is universal with respect to both quantum and conventional forms 
of computation, it is important to consider how well an
approximation to the ground-state energy can be obtained by purely classical
combinatorial optimization techniques. Consequently, Bansal et al. \cite{BansalBT09}
proposed an approximation algorithm for GSD on Ising spin lattices, {\em which
essentially simulates these AQC architectures \cite{KaminskyLO4, Dwave}, and
thus limits their potential for quantum speed-ups}. 
To approximate the least energy with $\epsilon$ accuracy,
the algorithm from \cite{BansalBT09} requires runtime exponential in
$1/\epsilon$, which is hardly practical. In contrast, we propose a
branch-and-bound algorithm and a high-performance local search that quickly
finds near-optimal energy values for arbitrary Ising topologies. Such techniques 
can be used to study properties of solid-state
materials, as well as critically assess the performance of
non-traditional computing devices based on energy minimization in Ising
spin-glasses. The main contributions of our work are summarized as follows.
\begin{itemize}
 \item A branch-and-bound algorithm for solving GSD exactly
       on Ising systems with up to $100$ spins.
 \item A high-performance local search algorithm for Ising
       spin-glasses. Empirical results show that it scales better than other GSD
       algorithms and produces near-optimal solutions for small- to medium-sized
       instances.
 \item A generalization of GSD for simulating energy minimization in physical systems. 
In particular, we propose a self-contained number-factoring algorithm based on this approach. These results can be used as a baseline for evaluating the performance of non-traditional computing devices that solve hard problems via energy minimization (e.g. AQC).
 \item A comparison of two potential spin-glass architectures for AQC number factoring.
\end{itemize}
The rest of the paper is structured as follows. In Section \ref{sec:prev} we
discuss common variants of Ising models, as well as the best known algorithms
for solving GSD. Sections \ref{sec:exactgs} and \ref{sec:heuristicgs} introduce
our algorithms for finding ground states.  Section \ref{sec:results} reports
empirical results for calculating ground states. We build upon these findings
and describe our generalization of Ising models using hypergraphs in Section
\ref{sec:factor} and compare two potential architectures for AQC integer 
factoring. We finalize the discussion with concluding remarks in Section
\ref{sec:conclude}.

\section{Background and Previous Work}  \label{sec:prev}

Due to its flexibility, the Ising model has been reinterpreted by researchers
to analyze different kinds of physical systems. Potts \cite{Wu82} suggested a
generalization of the model where the spin values are uniformly distributed
about the unit circle. In the Edwards-Anderson (EA) \cite{EdwardsA75} interpretation,
the spins are binary $\pm 1$ values, and the strengths of the atomic couplings
are independent, identically distributed random variables. Ising systems that share this property are known
as {\em spin glasses} since the random positive/negative edge-weights simulate
the solid-state structure of chemical glass. Let $G_{ising}=(V,E)$ denote a spin-glass graph with $n$ vertices.
Each vertex $u \in V$ is annotated with spin value $S_u \in \{\pm 1\}$
and is assigned a magnetization weight $h_u$. For $u, v \in V$, define $(u,
v) \in E$ to be an edge representing a bond between two adjacent spins with
assigned weight $J_{u,v}$ chosen randomly from either the standard Gaussian
($\mu = 0, \delta = 1$) or the $\pm1$-bimodal distributions. The internal energy of the system for a
particular configuration of spin values $\sigma = \{S_i\}$ is given by
    \begin{equation} \label{eq:energy}
        E(\sigma) = -\displaystyle\sum_{(i, j) \in E}^n J_{i,j}S_iS_j - \sum_ih_iS_i
    \end{equation}
\noindent where the summation considers all pairs of adjacent spins. Putting
together the energies of all spin configurations gives the {\em Hamiltonian} of
the system. Thus, the ground state is given by $E_{gs} = \min(E(\sigma)$ $\mid$
$\forall$ $\sigma\in\pi_n)$, where $\pi_n$ is the set of all possible $n$-spin
configurations. Whether we are interested in the lowest-energy value or the
$n$-spin configuration with such energy, $\mid\pi_n\mid = 2^n$ because
each of the spins can take on one of two possible values. As discussed in the
next section, energy minimization is typically NP-hard. Thus, calculating
the ground state exactly using an exhaustive search algorithm is feasible only
for small Ising models. To provide a scalable way of finding ground states or
approximating their energies, we need to employ heuristics such as those
proposed in Section \ref{sec:heuristicgs}. First, consider the following
definitions. 

    \begin{definition}
	A bond is {\em satisfied} if and only if the configuration of its
	incident spins minimizes its weight (coupling strength) such that
	$-J_{i,j}S_iS_j = -\mid J_{i,j}\mid$; otherwise, the bond is {\em unsatisfied}.
    \end{definition}


    \begin{definition}
	A set of spins $S \subseteq V$ is {\em frustrated} if there is no
	configuration of the spins that satisfies all the bonds $(u, v)$ connecting the
	spins in the set, i.e., $v, u \in S$.
    \end{definition}

    \begin{table*}[!htb]
        \begin{center}
            \begin{tabular}{|c|c|c|c||c|c|}
                \hline
                Lattice    & Boundary   & External      & Bond-weight & \bf NP-hard? & Poly-time \\
                dimensions & conditions & magnetization & signs       &              & algorithm \\ \hline \hline
                $1$ & $\leq 1$ & Yes/No & $\pm$ & \bf No & Analytical \\ \hline
                $2$ & $0$ & No & $\pm$ & \bf No & MWPM \\ \hline
                $2$ & $0$ & Yes & $\pm$ & \bf Yes & -- \\ \hline
                $2$ & $\leq 2$ & Yes & $+$ & \bf No & Max-flow \\ \hline
                $2$ & $1$ & No & $\pm$ & \bf No & MWPM \\ \hline
                $2$ & $1$ & Yes & $\pm$ & \bf Yes & -- \\ \hline
                $2$ & $2$ & Yes/No & $\pm$ & \bf Yes & -- \\ \hline
                $N>2$ & $\leq N$ & Yes/No & $\pm$ & \bf Yes & -- \\ \hline
            \end{tabular}
            \vspace{-10pt}
		\parbox{10cm}{
            \caption{ \label{tab:gsdnpc}
			Ising spin-glass properties that make the GSD problem NP-hard on lattices. 
			MWPM stands for minimum-weight perfect matching.}
		}
        \end{center}
    \end{table*}

\noindent Alternatively, the ground state is defined by the spin configuration
that minimizes frustration in the Ising system, provided that the system is not
affected by an external magnetic field. That is, in a spin configuration without
any frustrated spin-sets, all bonds have been satisfied and the energy of the
system is $E_{gs} = -\sum_{(i,j)\in E}\mid J_{i,j}\mid$, which is a lower bound of the
energy function.

In general, Ising-model graphs are not limited to a particular topology, but
two- and three-dimensional lattices are most commonly considered in the
literature. To simulate the behavior of infinite spin glasses, it is common to
require that the spins lying on the dimensional boundary be connected to the
spins on the opposing boundary on the same dimension. This can be viewed as a
type of (periodic) boundary condition. In particular, only one periodic
boundary condition is imposed for each dimension of the lattice. However, it is
sometimes desirable to have boundary conditions on some but not all of the
lattice dimensions. For example, a $2$-D lattice may have zero
(planar grid), one or two boundary conditions. When no boundary conditions are
imposed, some (boundary) spins have fewer than four neighbors. \\

\noindent {\bf Complexity of GSD}. In his work, Barahona \cite{Barahona83} showed that
the NP-complete task of finding a maximum set of independent edges (edges
with no common incident vertices) in a graph can be reduced to GSD on a cubic grid. 
Although most variations of GSD are known to be NP-hard, there
are a few cases where the structure of the graph can be exploited to solve the
problem in polynomial time. For example, Bieche et al. \cite{Bieche80} proved that
the GSD problem on planar graphs with zero magnetization can be solved in polynomial time by showing a reduction to the minimum-weight perfect matching (MWPM) problem. It follows
from their work that GSD instances with zero magnetization ($h_i = 0$) and $0$-
or $1$-periodic boundary conditions can be solved in $O(n^3)$ time \cite{Edmonds65}.
Specifically, the algorithm identifies the frustrated faces in a grid
($4$-cycles that have an odd number of negative edges) as vertices in a new
graph $G_F = (F, E_F)$. $G_F$ is complete and each edge $e = (f_i, f_j) \in
E_F$ is assigned a weight equal to the sum of the absolute weights of
the edges in the original graph that are crossed by the minimum path that connects
$f_i$ and $f_j$. Recall that the edges in $G_F$ are connecting sets of
frustrated spins. Therefore, minimizing the sum of the weights connecting $f_i
\in F$ implies that we are minimizing frustration (Definition 3). Thus,
finding the ground state is reduced to finding a MWPM on $G_F$. However,
although MWPM is solvable in polynomial time, the runtime is impractical for
large instances and suffers from a big memory footprint due to the size of
$G_F$. To overcome these limitations, the work in \cite{PardellaL08} describes a
heuristic based on the MWPM reduction where a reduced graph $\tilde{G}_F$ is
used instead of a complete one at the cost of sub-optimality.

Table \ref{tab:gsdnpc} shows the Ising lattice properties that make a GSD
problem poly-time solvable and identifies the algorithms that are commonly
used. Note that the number of dimensions, the number of boundary conditions
and the presence of an external magnetic field are the main factors in
determining whether an instance is NP-hard or not. More precisely, when we
consider lattices with more than two dimensions or with two boundary
conditions, the graph is no longer planar and the reduction to MWPM breaks
down. Another special case considered in the literature is that of
ferromagnetic ($J_{i,j} > 0$) GSD instances. Barahona \cite{Barahona94} reduced this
particular problem to ($s$-$t$)-min-cut or max-flow. 




\section{Finding Exact Ground States}   \label{sec:exactgs}

To better control the trade-offs between runtime and solution quality of
heuristics, it is important to design algorithms that are guaranteed to find
exact ground states on smaller instances. The solutions obtained from such
instances are used to evaluate scalable heuristics. \\

\noindent {\bf Branch-and-bound (B\&B)}. For general optimization problems,
B\&B considers incomplete or {\em partial} solutions, where only a
subset of the problem variables are assigned admissible values. Partial
solutions are systematically constructed via {\em branching}. The
branching process only develops partial solutions that are deemed promising,
i.e., those that may lead to the optimal solution. Conversely, partial
solutions whose cost is too high, are ``bounded away'' or pruned. \\

\noindent {\bf B\&B on Ising systems}. Our B\&B algorithm proceeds as follows. First, all spins are labeled as unassigned--their value can be set in the future to either $1$ or $-1$. The algorithm then calculates the lower bound of Equation \ref{eq:lbe},
    \begin{equation}
        E_{lb} = -\sum_{(i, j) \in E}^n \mid J_{i,j}\mid - \sum_i\mid h_i\mid
        \label{eq:lbe}
    \end{equation}
\noindent It then selects a spin $i$ and branches on one of the possible values for the spin that has not been explored yet. In each branch, the incremental change in $E_{lb}$ caused by the assignment is recorded as follows. For each spin $j$ adjacent to $i$ that has already been assigned, increase (decrease) $E_{lb}$ by twice the amount of the positive (negative) bond connecting $i$ and $j$ if they have opposing (aligned) spin values,
    \begin{equation}
        E^\delta_{lb} = \begin{cases}
                        2\displaystyle\sum_{(i, j) \in E}J_{i,j}S_iS_j & \text{if $S_i \neq S_j$ and $J_{i,j} > 0$,} \\
                        -2\displaystyle\sum_{(i, j) \in E}J_{i,j}S_iS_j & \text{if $S_i = S_j$ and $J_{i,j} < 0$}
                 \end{cases}
        \label{eq:edelta}
    \end{equation}
\noindent Furthermore, the corresponding change due to the magnetization of the spin is also recorded,
    \begin{equation}
        E^\delta_{lb} = \begin{cases}
                            2h_i & \text{if $S_i = -1$ and $h_i > 0$,} \\
                            -2h_i & \text{if $S_i = +1$ and $h_i < 0$}
                 \end{cases}
        \label{eq:edeltam}
    \end{equation}
\noindent Once the spin is assigned, the algorithm branches out to another spin and performs the same procedure. When all spins have been assigned, $E_{lb}$ represents the energy of the spin configuration generated by the branching process. To continue searching the configuration space, the branching process backtracks to the last assigned spin, flips the spin's value and updates $E_{lb}$. If both spin values have already been tried, then the algorithm continues backtracking while relabeling the spins as unassigned. Since each spin can take one of two values, this branching process generates a full binary search tree where the leaves correspond to all possible spin configurations in the Ising system.

Initially, we use a linear-time greedy approximation ($E_{gs}$) as our bounding
value. During the branching process, if the energy of the partial solution
exceeds $E_{gs}$, then we can safely prune this branch and backtrack to the
previous assigned spin without making any further assignments. The algorithm
either tries the opposite spin value or backtracks again if both spin values
have already been tried. If the search assigns all the spins in the graph and
the corresponding minimal energy state is lower than $E_{gs}$, then we set
$E_{gs}$ to this new energy value. After searching all promising branches,
$E_{gs}$ will assume the ground-state energy. This standard bounding technique
alone improves the scalability of the branching process by an order of
magnitude over exhaustive search (see Figure \ref{fig:bbchart}).

    \begin{figure}[b]
	\hspace{-5mm}
	\includegraphics[width=\linewidth, height=58mm]{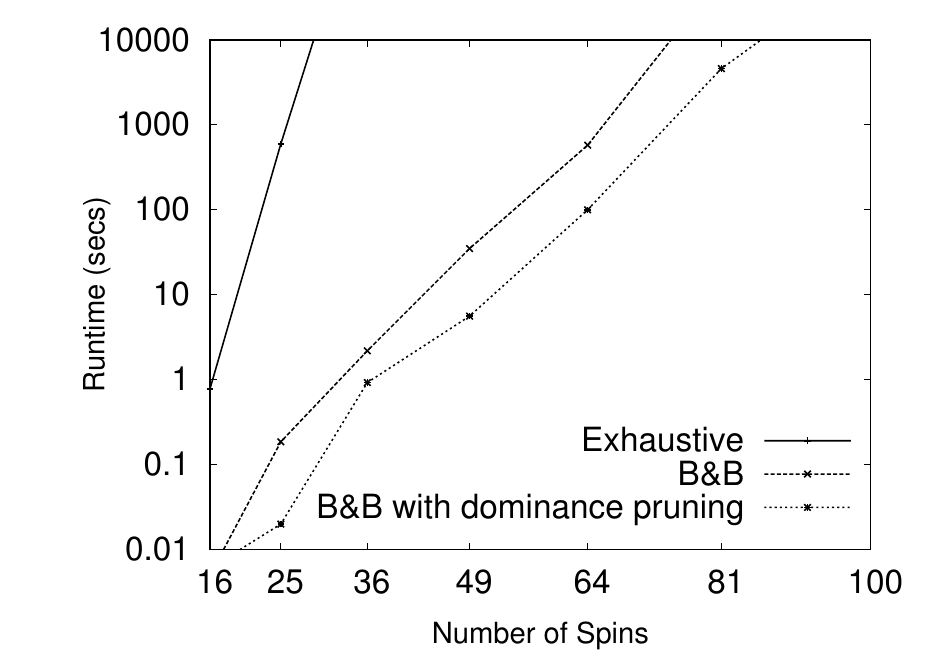} 
        \parbox{8cm}{
	 \caption{\label{fig:bbchart}Performance of B\&B techniques on $2$-D spin lattices.}}
    \end{figure}

To further improve the scalability of our B\&B algorithm, we designed a {\em
prune-by-dominance} technique that consists of identifying partial solutions
whose partial energy can be improved (lowered) by modifying the configuration
of currently assigned spins. Note that, whenever we assign a spin $s$, there is
a set $F_s$ of spins adjacent to $s$ for which all neighboring spins (including
$s$) have also been assigned. Early in the branching process, $F_s$ is
likely to be empty since only a few spins have been assigned. As spins 
are assigned, the set $F_s$ increases. Figure \ref{fig:bbdom} shows an
example of $s$ and $F_s$ on a small grid. The size of $F_s$ is no greater than
the degree of $s$. Note that the energy of the spins in $F_s$ is localized in
the sense that it will not be affected by further spin assignments. 

    \begin{lemma} \label{lem:bblocale}
	Let $F_s$ be the set of spins such that $\forall i \in F_s$, all spins
	adjacent to $i$ are assigned. Then the partial energy of the spins in $F_s$
	will not be affected by additional spin assignments.
    \end{lemma}

    \begin{proof}
	Let $i \in F_s$, then the partial energy lower bound that is localized
	around spin $i$ is the given by $E^i_{lb} = \pm 2\sum_{(i, j) \in
	E}J_{i,j}S_iS_j \pm 2h_i$ (the $\pm$ stands for the cases described in
	Equations \ref{eq:edelta} and \ref{eq:edeltam}). By definition of $F_s$, we
	know that every spin $j$ adjacent to $i$ has also been assigned. Suppose that
	later in the branching process we assign spin $u$ and this causes a change in
	$E^i_{lb}$. By definition of $E^i_{lb}$, $u$ must be adjacent to $i$.
	This implies that $i \notin F_s$ since its neighbor $u$ had not been assigned
	until recently, which is a contradiction.
    \end{proof}

This allows us to evaluate a partial solution by flipping the values of the spins in $F_s$ and
comparing the partial energies. If any of the partial energies are lower than
$E_{lb}$ or if the modified configuration corresponds to a visited partial
solution, then we know that the current partial solution is unpromising and we
can safely backtrack.

    \begin{figure}[ht]
        \centering
        \includegraphics[width=.5\linewidth, height=25mm]{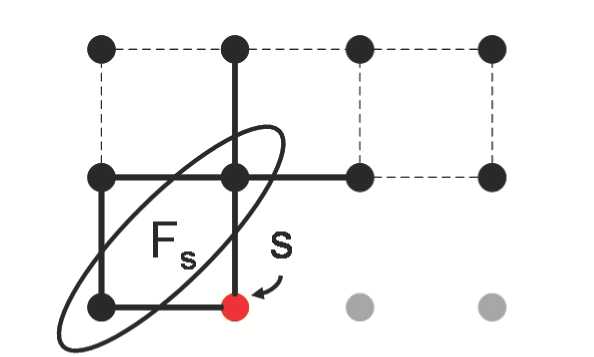}
        \caption{\label{fig:bbdom} Illustration of $s$ and $F_s$ on a small grid.}
    \end{figure}

    \begin{lemma} \label{lem:bbprune}
	Let $F_s$ be the set of spins such that $\forall i \in F_s$, all spins
	adjacent to $i$ are assigned. Also, let $\sigma_{F_s}$ correspond to some
	configuration of the spins in $F_s$. If we can find $\sigma^\prime_{F_s}$ such
	that $E(\sigma^\prime_{F_s}) < E(\sigma_{F_s})$, then any branches that include
	$\sigma_{F_s}$ can be safely pruned.
    \end{lemma}

    \begin{proof}
	Consider the partial energy $E(\sigma_{F_s})$ and the total energy of
	any complete spin-configuration $\sigma$ that extends the partial
	spin-configuration. Let $\sigma^\prime_{F_s}$ be a configuration of $F_s$ that
	minimizes the partial energy, i. e., $E(\sigma^\prime_{F_s}) <
	E(\sigma_{F_s})$. By Lemma \ref{lem:bblocale}, new spin assignments will not
	affect $E(\sigma_{F_s})$. Since $\sigma_{F_s}$ is included in $\sigma$ we can
	swap $\sigma_{F_s}$ with $\sigma^\prime_{F_s}$ so that $E(\sigma)$ is also
	minimized. Thus, any partial or complete spin-configuration that includes
	$\sigma_{F_s}$ is not promising. \\
    \end{proof}

Observe that in cases when different branches are
unlikely to have equal partial cost (e.g., when couplings and
magnetizations are random), for two branches, the probability that the first
branch dominates the second branch is approximately $1/2$. Let $0 < c \leq 1$
be the fraction of ($2^k$ partial solutions) that require branching. Then we
can expect  to prune $c(2^k/2) = c(2^{k-1})$ of these branches.
As seen in Figure \ref{fig:bbchart}, this pruning technique improved the
scalability of our B\&B algorithm by 1-2 orders of magnitude,
allowing it to solve $100$-spin lattices in a day. However, even with
the techniques proposed, B\&B takes exponential time in the worst
case and therefore fails to scale beyond $100$ spins.

\section{GSD through Local Search} \label{sec:heuristicgs}

Due to the difficulty of solving general instances of GSD, many
researchers \cite{Bieche80} \cite{PardellaL08} have developed heuristic methods to improve the
scalability of their algorithms at the expense of solution quality, typically
based on slow Monte Carlo simulations. However, because of the role that Ising
models play in simulating real-world phenomena, it is desirable to have much
faster techniques. To this end, we propose a high-performance local search that
meets such scalability and performance requirements.

Our local search is an iterative improvement algorithm that modifies the bipartition induced by an arbitrary spin
configuration (positive spins are placed in one partition and negative spins in
the other). The algorithm performs a sequence of incremental changes to the
bipartition, organized as {\em passes}. These changes consist of {\em spin
moves} that place a particular spin in the partition opposite to the one it
is currently in. At the beginning of each pass, the energy differential
({\em gain}) of performing each possible move is calculated. A positive gain
implies that the move decreases the overall energy while a negative gain
increases it. During a pass, the move that produces the largest gain is
selected and executed. The corresponding spin is then labeled as {\em locked},
i.e., it cannot be selected again in the current pass to prevent ``undo''
moves. The pass continues selecting and executing the best moves until all
spins have been locked. At the end of the pass, we save the best-seen
bipartition produced by the sequence of moves. This bipartition is then used as
the starting solution of the next pass. The entire algorithm terminates when a
pass fails to obtain an improvement in energy as shown in Figure
\ref{alg:locsearch}. Note that, in the absence of positive-gain moves, a
negative-gain move can be selected. Thus, a pass may accept a solution that is
worse than the existing solution (hill-climbing). This helps to reduce the
probability of getting trapped in local minima. Figure \ref{fig:fmpass} illustrates the progress of our local search in terms
of solution costs during individual passes on a $1024$-spin glass. The initial
random solution used in our local search is generated in linear time. Assuming that initial
spin configurations are drawn from a uniform distribution, our local search
finds an optimal solution with probability at least $1/2^n$ for $n$ spins (in
practice, much higher than that, as indicated in Figure \ref{fig:fmfactor}).
The speed of our algorithm can be converted into better solutions by generating
independent random inital spin configurations, running (otherwise deterministic)
optimization passes on each, and selecting the best result. \\

    \begin{figure}[!b]
        \centering
        \begin{tabular}{|c|}
        \hline
        \begin{minipage}{\linewidth}
        \footnotesize
        \vspace{1mm}
        \begin{program}
                |\bf Input: | |Ising spin-glass graph | G_{ising}
                |\bf Output: | |Approximate ground-state energy | E^*
                ---------------------------
                |Spin Partition | SP^* := RAND\_SPIN\_PARTITION(G_{ising})
                |\bf while| | solution quality improves | |\bf do|
                        ~~~~~~|Gains container | GC := COMPUTE\_GAINS(SP^*)
                        ~~~~~~|Spin Partition | SP := SP^*
                        ~~~~~~|\bf while | GC | has unlocked spins | |\bf do|
                                ~~~~~~~~~~| Move | m := SELECT\_BEST\_MOVE(GC)
                                ~~~~~~~~~~APPLY\_MOVE(SP,m)
                                ~~~~~~~~~~UPDATE\_GAINS(GC, m)
                                ~~~~~~~~~~LOCK\_SPIN(GC, m)
                                ~~~~~~~~~~|\bf if| | energy decreased | |\bf then | SP^* := SP
                        ~~~~|\bf end while |
                |\bf end while |
                |\bf return | ENERGY(G_{ising},SP^*)
        \end{program}
        \end{minipage} \\ \hline
        \end{tabular}
        \caption{\label{alg:locsearch} Pass-based local search with hill-climbing}
    \end{figure}

\noindent {\bf Efficient gain updates.} Each move causes a change in the 
local energy surrounding the selected spin, therefore, the gains of the 
neighboring spins need to updated after each move.
When the graph capturing a spin system is sparse (e.g., Ising lattices),
only a constant number of gain updates are executed per move.
These gain updates are performed efficiently using the custom heap-based data
structure shown in Figure \ref{fig:heapgains}. The data structure consists of
two arrays. The first array implements a traditional binary heap while the
second array allows quick access to the heap-array element that contains the
gain-update value of a particular spin. To perform gain updates, we can access
the specific value in $O(1)$ time, update the value, and perform the necessary
swaps to maintain heap-order property. Since only $\log(n)$ swaps are required
in the worst case (where $n$ is the number of spins), our data structure allows
us to perform gain updates much faster than naive implementations that require
scanning the entire set of $n$ gain values. Since a total of $n$ moves are 
performed during a pass. This gives a total runtime of $O(n\log(n))$
for a single pass.\\

     \begin{figure}[!t]
        \centering
        \includegraphics[width=.85\linewidth, height=35mm]{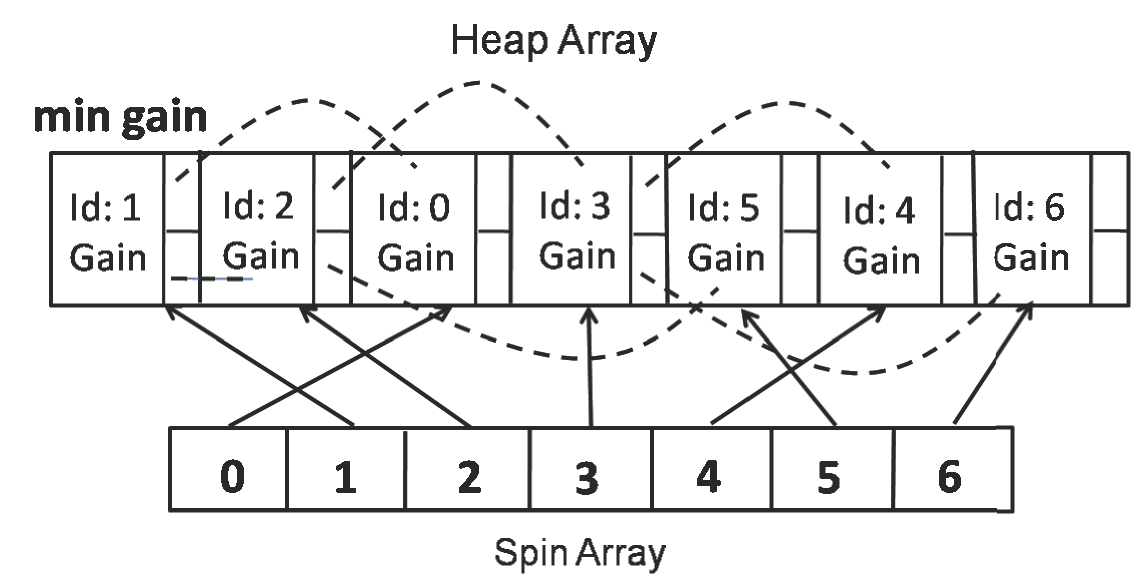}
	\caption{\label{fig:heapgains} Heap-based data structure for performing gain updates
	efficiently. Dashed lines indicate pointers to child nodes.
	Spin array pointers are updated accordingly after every heap swap.}
    \end{figure}
 
    \begin{figure}[!b]
        \centering
        \includegraphics[scale=.6]{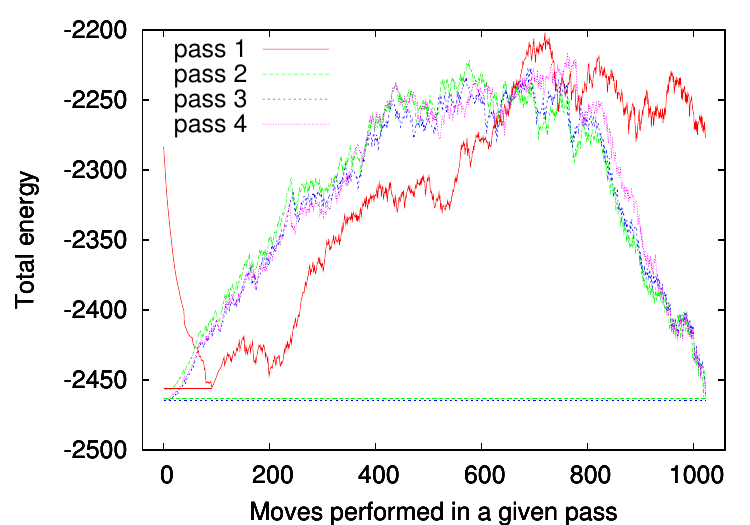}
	\caption{\label{fig:fmpass} Progress of local search in terms of total
	energy during individual passes on a $1024$-spin glass. The lowest-energy state
	observed during each pass ($E_1=-2456$, $E_2=-2463$, $E_3=-2465$, respectively)
	becomes the starting state of the next pass.}
    \end{figure}

\section{Empirical Validation} \label{sec:results}

We evaluated single-threaded implementations of our algorithms on a
conventional Linux server, although our local search is trivial to parallelize
to a multicore system or a distributed cluster. For $15 \times 15$ spin
lattices our local search finds exact ground states in $95\%$ of independent
random starts (exact solutions were obtained from \cite{UCologne}), otherwise
solutions are $5\%$ sub-optimal on average. Figure \ref{fig:fmqualchart}(a)
compares the average solution quality of local search for $2$-D spin glasses
with Gaussian-distributed couplings and $h_i = 0$ (instances with $h_i\neq 0$
are not allowed in \cite{UCologne}). For each instance we considered four
different levels of effort with an increasing number of independent random
starts. To obtain the average solution quality we computed $1000$ output
samples using $1$, $\ln^2 n$ and $n$ random starts ($n$ is the number of spins)
per instance. For $n^2$ random starts, we used fewer output samples and provide confidence intervals.
As expected, the solution quality improves as the number of random starts
increases. When at least $\ln^2 n$ random starts are used, our heuristic
produces high-quality solutions ($> 95\%$) for five of the benchmarks while its
runtime does not exceed $17$ seconds for the largest benchmark ($2500$ spins).
Note that the expected solution quality slowly decreases for larger instances.
Figure \ref{fig:fmqualchart}(b) shows similar results for benchmarks with
$\pm1$-bimodal coupling distributions, but solutions are closer to optimal. For
all but one of the benchmarks, our heuristic requires only a single random
start to find high-quality solutions. \\

    \begin{figure}[!b]
        \centering
        \includegraphics[width=\columnwidth, height=55mm]{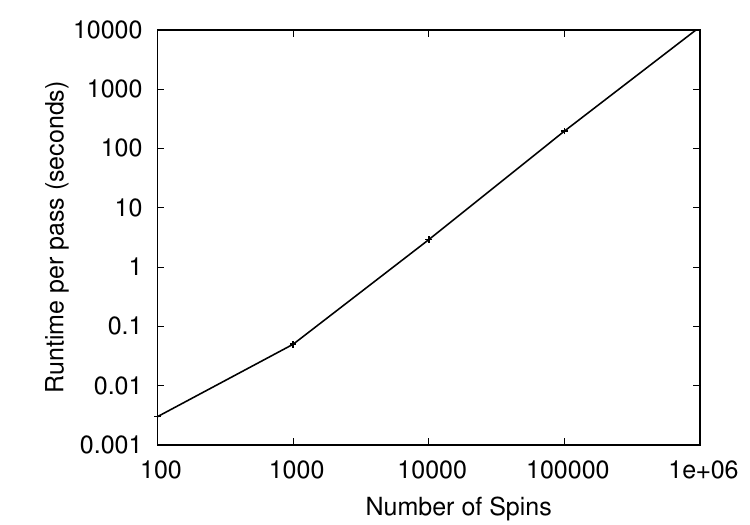}
\parbox{8cm}{
        \caption{Runtime of local search on large $2$-D spin lattices.}
        \label{fig:runtimes}
}
    \end{figure}

    \begin{figure*}
        \centering
        \begin{tabular}{cc}
            \includegraphics[width=\columnwidth]{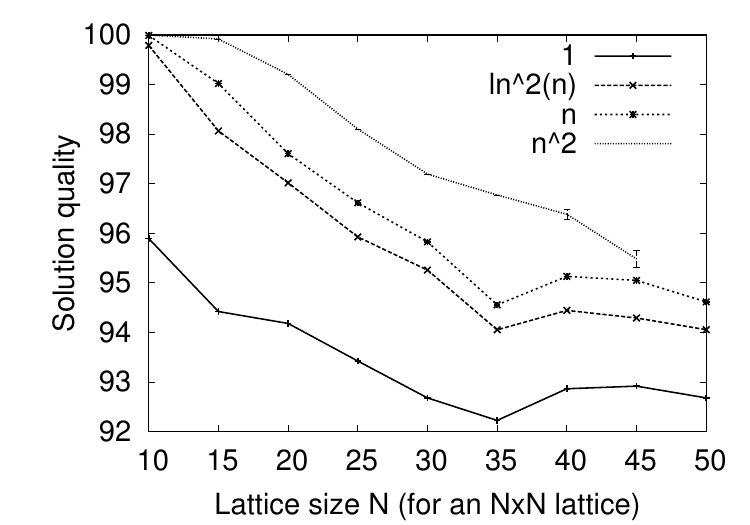} 
            & \includegraphics[width=\columnwidth]{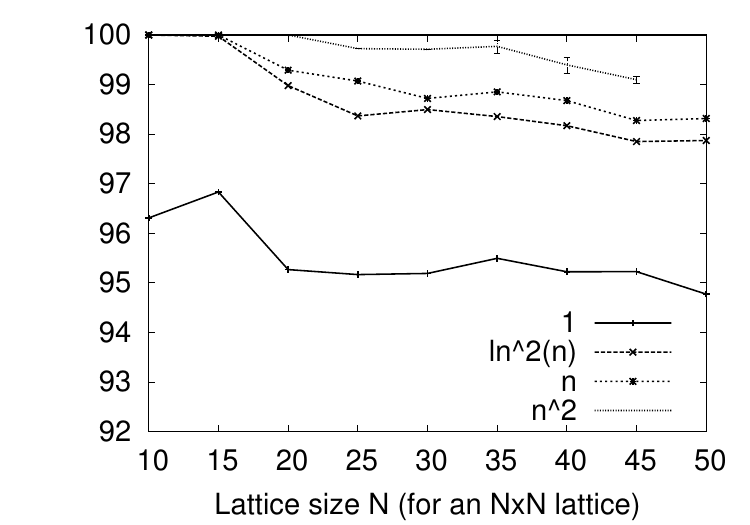} \\
            (a) \footnotesize{Gaussian-distributed couplings} & (b) \footnotesize{Bimodal-distributed couplings} \\
        \end{tabular}
\parbox{12cm}{
        \caption{\label{fig:fmqualchart} Expected solution quality ($100\%$ corresponds to the exact solution obtained from \cite{UCologne}) produced by using different numbers of random starts on $2$-D spin lattices.}
             }
     \end{figure*}

\noindent {\bf Local optimality}. We verified that the
configurations returned by our heuristic cannot be improved by modifying a
small number of spins. We used our B\&B to find optimal configurations of
groups of $25$-$49$ adjacent spins within larger configurations. In our
experiments, the solutions produced by our local search were never improved by
this technique. \\

\noindent {\bf Runtime}. Figure \ref{fig:runtimes} shows that our heuristic
scales to a million spins and its runtime is consistent with the complexity
estimate of $O(n\log(n))$ per pass (see Section \ref{sec:heuristicgs}). We compared the
runtime of our local search against that of the MWPM-based heuristic proposed
in \cite{PardellaL08} (see Section \ref{sec:prev}). Recall that the
heuristic works on a reduced dual graph instead of the complete one required to
find exact ground states. The reduced dual graph ignores those edges that are
deemed ``heavy'', i.e., with weights above a chosen threshold $c_{max} = c\cdot
J_{max}$, where $c = 2,4,6$. The idea is that these heavy edges are rarely
contained in the optimum solution and can be ignored. Thus, the solution
quality of the heuristic depends on $c_{max}$ and the positive-to-negative edge
ratio since the more edges we ignore, the higher the probability of missing the
ground state. In contrast, our heuristic does not have such a dependency and
works on all instances. Furthermore, it is not clear how to chose $c_{max}$ in
the case that the couplings weights are Gaussian distributed. In contrast, our
heuristic can be used with any coupling distribution. Table IV in \cite{PardellaL08}
shows the runtime and solution quality of the MWPM-based heuristic for
different values of $c_{max}$ on $\pm1$ grid graphs of size $164 \times 164$
with $.5$ positive-to-negative edge ratio. The fastest runtime of the cited
heuristic is obtained when $c_{max} = 3$ taking an average of $58.72$ seconds
(with negligible deviation) and producing the optimal value only $61\%$ of the
time. By comparison, our local search heuristic on a comparable benchmark with
a single random start takes about $8.5$ seconds. Thus, we can perform $7$
random starts in the same period of time. However, we have no way of comparing
the quality of our solutions since we do not have access to the same benchmarks.
Section 4.5 of \cite{HartmanR04} describes the branch-and-cut approach used by the
Cologne Spin Glass Server. This algorithm is limited to lattices
without magnetization (our techniques do not suffer this limitation) 
and mentions a $128$-second runtime for the largest instance ($50 \times 50$).
In contrast, our local search heuristic takes only an average of $16.63$
seconds using $\ln^2(n)$ random starts on the same instance and produces
solutions that are within $95.5\%$ of the optimal value on average. Solutions
closer to optimum may be found if the number of random starts is increased.

\section{GSD for arbitrary hypergraphs} \label{sec:factor}

As discussed in Section \ref{sec:intro}, recent AQC architectures \cite{Peng08}
are based on Ising spin glasses. Implementations described by D-Wave Systems use
non-planar topologies, and recent experiments in \cite{Kim09} demonstrate
direct coupling of more than two spins. Hence, we extend the conventional spin-glass 
model to use {\em hyper-couplings} that connect a set of at least two
spins. The new energy function is given by
     \begin{equation} \label{eq:hyperenergy}
        E(\sigma) = -\displaystyle\sum_{e \in E} J_{e}\prod_{i\in e}S_i - \sum_{i \in V}h_iS_i
     \end{equation}
This formulation further extends the applicability of our computational approach to
energy minimization. Adapting our algorithms to handle this formulation allows
us to study a greater variety of physical systems and solve a wider 
range of combinatorial problems (e.g., number fac\-tor\-ing--discussed later
in this section). Our algorithms required minor modifications to handle hyper-Ising models. In the
case of B\&B, the incremental energy calculations (Equations
\ref{eq:edelta} and \ref{eq:edeltam}) were updated to conform with Equation \ref{eq:factore}. For local
search, the gain-update procedure was revised to account for the presence of
hyper-couplings. \\

\noindent {\bf Number-factoring as an optimization problem}. We observe that integer factorization is equivalent to the optimization of the following function,
    \begin{equation}    \label{eq:factore}
        f(x,y) = (N - xy)^2
    \end{equation}
\noindent where $N$ is the odd integer we wish to factor and $x$, $y$ are positive integers. The minimum of $f(x,y)$ is reached when $x$ and $y$ are factors of $N$. To solve this optimization problem using GSD, we must construct an Ising system whose ground state corresponds to the global minimum of $f(x,y)$. Since spins are $\pm1$ binary variables, we re-formulate $f(x,y)$ in terms of binary digits. Let $n_x$ and $n_y$ be the number of binary digits required to represent $x$ and $y$, respectively.
    \begin{equation}    \label{eq:factorbin}
        f(x,y) = \left[N - \left(\sum_{i = 1}^{n_x - 1} 2^{i}x_i + 1\right)\left(\sum_{i = 1}^{n_y - 1} 2^{i}y_i + 1\right)\right]^2
    \end{equation}
\noindent Since $N$ is odd, $x$ and $y$ are also odd and therefore the last term in each sum is one. Thus, the variables $x_0$ and $y_0$ are ignored. The number of spins is $(n_x - 1) + (n_y - 1) = \mid V\mid$. Setting $S_k = 2x_k - 1$ and $S_{n_x + k} = 2y_k - 1$,
    \begin{align}    \label{eq:factorising}
	\hspace{-1mm}
        f(x,y) &= \left[N-\left(2^{n_x - 1}\frac{1 - S_{n-1}}{2} + ... + 2\frac{1 - S_{n - n_x}}{2} + 1\right)\right. \notag \\
        &\left.\left(2^{n_y - 1}\frac{1 - S_{n - n_x - 1}}{2} + ... + 2\frac{1 - S_{1}}{2} + 1\right)\right]^2
    \end{align}
\noindent The magnetization weights and coupling strengths are given by the
product expansion of Equation \ref{eq:factorising}. Terms of the form  $c\cdot
S_i$ correspond to spins with magnetization $h_i = c$, and terms of the form
$c\cdot S_iS_j..S_k$ correspond to hyper-couplings $e = (i, ..., k)$ with
strength $J_e = c$. Figure \ref{fig:factor} shows the Ising graph constructed
for factoring $21$. \\

    \begin{figure}[!b]
        \centering
        \includegraphics[scale=.5]{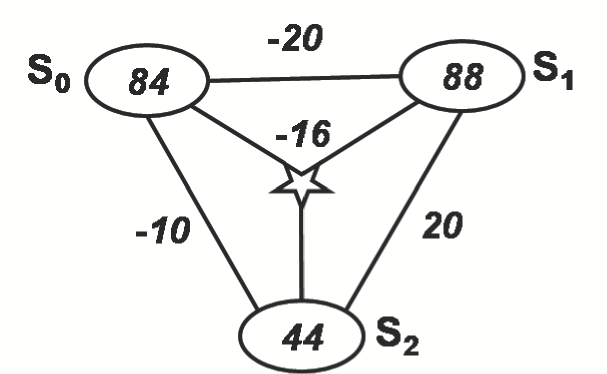}
        \caption{Ising-model hypergraph for factoring $N = 21$. In this example, we assume $x < y$ and set $n_x = 1$ and $n_y = 2$. The optimization function is $f(x,y) = 210+84S_0 + 88S_1 + 44S_2 -20S_0S_1-10S_0S_2+20S_1S_2-16S_0S_1S_2$.}
        \label{fig:factor}
     \end{figure}

    \begin{figure*}
        \centering
        \begin{tabular}{cc}
            \includegraphics[angle=-90, scale=.35]{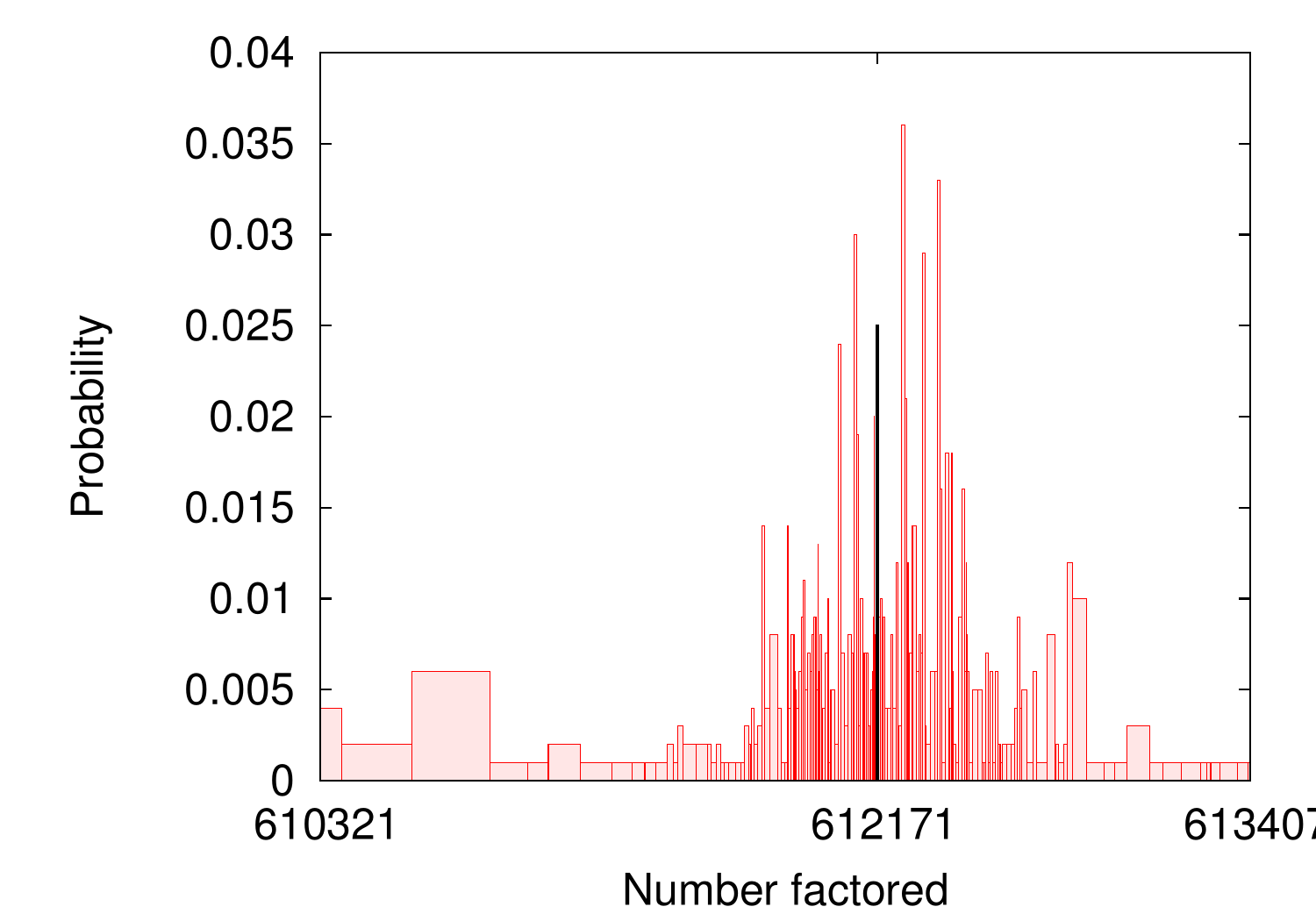} 
            & \includegraphics[angle=-90, scale=.35]{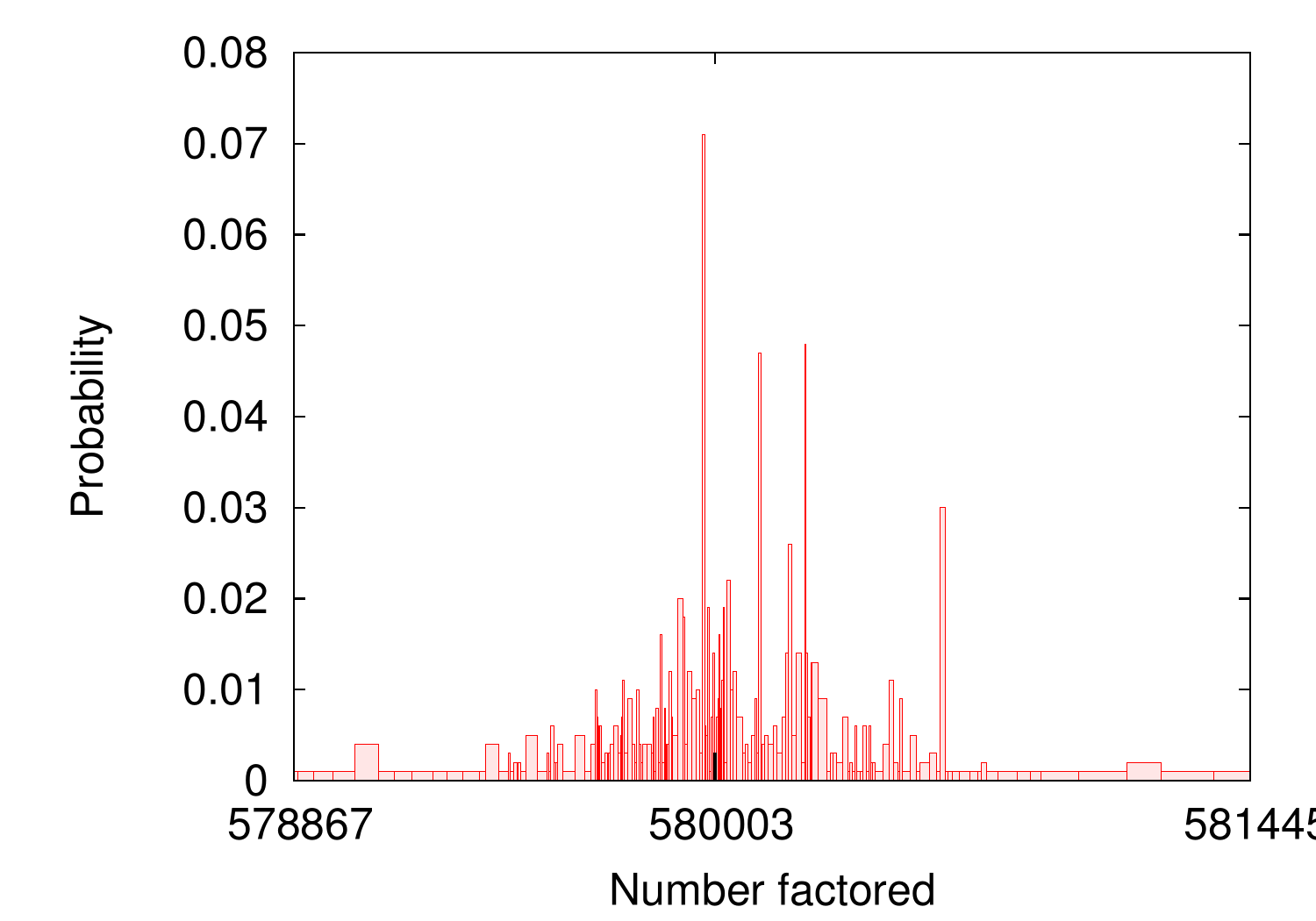}\\
            (a) \footnotesize{$612171$ is factored with probability $.025$.} & (b)\footnotesize{$580003$ is factored with probability $.005$.} \\
               Factors are $3$, $7$, $41$ and $79$. & Factors are $619$ and $937$. \\
        \end{tabular}
\parbox{15cm}{
        \caption{\label{fig:fmfactor}Probability distribution of factoring numbers using our GSD heuristic. The black lines show the probability of finding the correct factors using a single random start. Given that $100$ attempts take $<.01$ seconds, such numbers can be factored reliably.}
             }
    \end{figure*}

\noindent {\bf Computational experiments}. B\&B can factor up to $21$-bit
numbers in approximately $5$ minutes.  While leading-edge number-factoring
algorithms can do better, the results confirm that B\&B is general enough to
work on hyper-Ising models. We tested our GSD heuristic by factoring specific
numbers. In some runs, this technique produces the factors of $N\pm 2$ or
other incorrect numbers. Therefore, many independent random runs may be
required to factor a given number $N$, which is also typical behavior in 
adiabatic quantum computers (AQC). Figure \ref{fig:fmfactor} shows the output probability distribution of
factoring the number $612171$, which has several factors and is therefore easy
to factor. The probability of success drops sharply for semi-primes, e.g., the
probability of factoring $580003$ in one attempt is only $.005$. Therefore, our 
heuristic takes longer to find correct factors in such cases.
The current implementation of our algorithm is not particularly optimized, but
factored the semi-prime $10185081163 = 100511 \times 101333$ using $15,000$
random starts in $13$ seconds.  While further optimizations may significantly
improve runtime, present results provide a computational baseline for
performance evaluation of non-traditional computing devices that solve hard
problems using the principle of energy minimization. Such devices, e.g.,
AQC, would need to improve on our results by producing
distributions that increase the probability of factoring the correct number. \\

\noindent {\bf Significance of hyper-couplings}.
Recently, Peng et al. \cite{Peng08} implemented an AQC
number-factoring algorithm in NMR technology using a Hamiltonian that ignores
hyper-couplings. By simulating the functionality of the same algorithm, we replicated
their experiment and explored factorization of
larger numbers. As in \cite{Peng08}, the number $21$ was successfully factored
even in the absence of hyper-couplings. However, our experiments show that, in
general, omitting hyper-couplings alters the ground states so that correct
factors cannot be found, e.g., for $35$ and $91$. Instead, some other numbers
are factored such as $33$ and $95$. Although hyper-couplings have smaller
individual weights than two-spin couplings, the number of hyper-couplings
scales as $O(n^4)$, and their total weight eventually dominates $f(x,y)$ for
larger $N$. Control of three-spin hyper-couplings has recently been demonstrated
\cite{Kim09}, but only for adjacent particles, which would be insufficient for
number-factoring applications. Furthermore, Equation \ref{eq:factorising} still
requires four-spin hyper-couplings which, as current research suggests, are
difficult to realize experimentally. \\

    \begin{figure}[!b]
        \centering
        \includegraphics[width=\columnwidth]{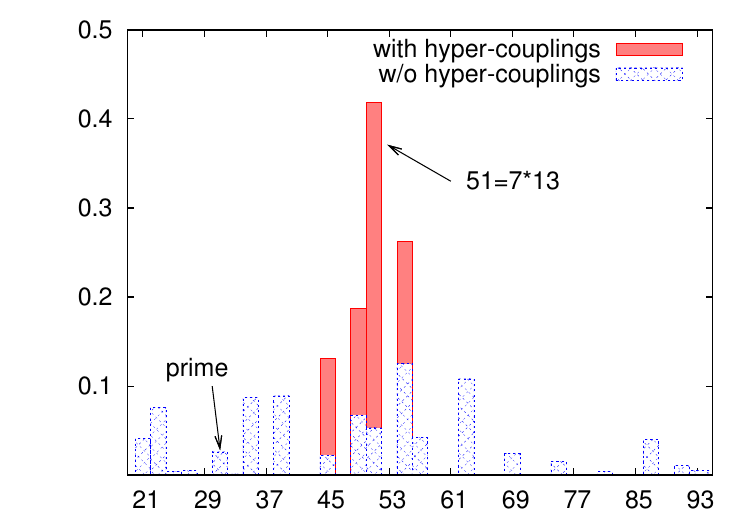} \\ 
		\parbox{7cm}{
        	\caption{\label{fig:aqccomp} Output probability distributions for factoring $51$ using the techniques 			from \cite{Peng08} and \cite{Schaller09} as simulated by our algorithms.}
        }
    \end{figure}

\noindent {\bf Avoiding hyper-couplings via ancilla spins}. The work in
\cite{Schaller09} shows a method for expanding Equation \ref{eq:factorising} that
avoids the use of hyper-couplings at the cost of a quadratic increase in spins.
The new Hamiltonian is based on the factorization equations that are generated by decomposing
long-hand binary multiplication \cite{Schaller09}. Equation \ref{eq:factorising} is modified
by introducing ancilla binary variables. Let $N=m_1m_2...m_n$. Let $x={x_1x_2...x_k}$
and $y={y_1y_2...y_{n-k}}$. Then let $p_{i,j}$ denote the sum of pairwise products between
the binary variables in Equation \ref{eq:factorising} and let $c_{i,j}$ denote carry variables
	\vspace{-3mm}
        \begin{align} \label{eq:factorisingnew}
        f(x,y) &= \displaystyle\sum_{i=1}^{k}\sum_{j=1}^{n-k}\left[2\left(
x_i/2 + y_j/2 + p_{i,j} + c_{i,j} \right.\right. \\
&\left.\left.-p_{i+1,j-1} - 2c_{i-1,j} - 1/4\right)^2 - 1/8\right] \notag
        \end{align}
	\vspace{-5mm}

\noindent The bits of $N$ are linked to the bits of $x$ and $y$ by a series of
equalities such as $p_{i,0}=m_i$ and $p_{k+1,j-1} = m_{k+j}$ (see
\cite{Schaller09} for details). $f(x,y)$ now computes a penalty function for
violating the factorization equations. This penalty is minimized subject to the
fixed values of $m_1...m_n$. Equation \ref{eq:factorisingnew} maps to a spin
system with only two-spin couplings. Unlike the (unconstrained) formulation
with hyper-couplings, this formulation includes spins with fixed values and
thus requires additional technology support (e.g., optical pumping of trapped
ions). We used our algorithms to compare the use of hyper-couplings (assuming
technological feasability) to the expansion from \cite{Schaller09}. Figure
\ref{fig:aqccomp} compares output probability distributions generated when
factoring $51$. The technique from \cite{Schaller09} produces a flatter
distribution with a wider range. This is because the solution space is more
complex and includes configurations with inconsistent ancilla spins (e.g.,
carry spin $c_{i,j}=0$, but the partial product spin $p_{l,k}=c_{i,j}*1=1$).
Thus, this technique sometimes returns trivial decompositions of prime numbers
(e.g., 31), whereas direct use of hyper-couplings always results in proper
factorizations. In summary, our experiments suggest that the expansion from
\cite{Schaller09} is not computationally competitive with direct use of
hyper-couplings.

\section{Conclusions} \label{sec:conclude}

The problem of finding the ground state of spin glasses (GSD) has played a
central role in statistical physics.  Additionally, high-performance GSD
algorithms can help evaluate architectural alternatives in adiabatic quantum
computing.  Recent work in \cite{Battaglia05} concludes that quantum
annealing loses to classical simulated annealing in a head-to-head
comparison, but can be improved by making different architectural choices.
Although our algorithms do not simulate quantum optimization directly, they
allow one to study problem reductions and identify potential obstacles to
successful optimization. The proposed B\&B algorithm can find ground states
for general $2$-D spin lattices with up to $100$ spins in $24$ hours. For GSD
instances with $1,000,000$ spins, our local search heuristic obtains
approximate solutions in $< 4$ hours. It provides a scalability advantage over
conventional Monte Carlo methods and is not limited to special classes of GSD
instances. When our heuristic does not find a ground state, it usually
approximates the least energy within $5\%$.

We replicated a recently published empirical result in AQC-based number
factoring \cite{Peng08}, where the omission of spin-spin hyper-couplings did
not undermine overall results.  However, we have shown that, in general,
omitting hyper-couplings will produce incorrect results. Furthermore, we
demonstrated that techniques avoiding the use of hyper-couplings
\cite{Schaller09} blur the output probability distribution, hamper finding
correct factors, and require more repeated attempts to achieve success.

\end{document}